\documentclass[12pt]{amsart}

\usepackage{amsmath,amsthm,amscd,euscript,longtable}
\def \r{\mathbb R}

\def \z{\mathbb Z}

\def \D{\mathcal D}

\def \({\langle}
\def \){\rangle}


\DeclareMathOperator{\sn}{sn} 
\DeclareMathOperator{\alt}{Alt}

\DeclareMathOperator{\crit}{Crit}

\newtheorem{theorem}{Theorem}[section]

\newtheorem{proposition}[theorem]{Proposition}
\newtheorem{corollary}[theorem]{Corollary}
\theoremstyle{remark}
\newtheorem{remark}[theorem]{Remark}
\theoremstyle{definition}
\newtheorem{definition}[theorem]{Definition}
\newtheorem{example}[theorem]{Example}
\newtheorem{problem}{Problem}
\newtheorem{conjecture}[problem]{Conjecture}

\title{energies of knot diagrams}
\author{Oleg Karpenkov, Alexey Sossinsky}
\date{2 June 2011}

\thanks{Oleg Karpenkov is partially supported by RFBR SS-709.2008.1
grant and by FWF grant No.~M~1273-N18.}
\thanks{Alexey Sossinsky is partially
supported by the RFBR-CNRS-a grant \#10-01-93111.}


\email[Oleg Karpenkov]{karpenkov@tugraz.at}

\email[Alexey Sossinsky]{asossinsky@yandex.ru}

\begin{document}
\input{epsf}

\begin{abstract}  We introduce and begin the study of new knot energies
defined on knot diagrams. Physically, they model the internal
energy of thin metallic solid tori squeezed between two parallel
planes. Thus the knots considered can perform the second and third
Reidemeister moves, but not the first one. The energy functionals
considered are the sum of two terms, the uniformization term
(which tends to make the curvature of the knot uniform) and the
resistance term (which, in particular, forbids crossing changes).
We define an infinite family of uniformization functionals,
depending on an arbitrary smooth function $f$ and study the
simplest nontrivial case $f(x)=x^2$, obtaining neat normal forms
(corresponding to minima of the functional) by making use of the
Gauss representation of immersed curves, of the phase space of the
pendulum, and of elliptic functions.
\end{abstract}

\maketitle

\section*{Introduction}

In this paper, we introduce certain energy functionals defined on
(a special class of) knot diagrams. Intuitively, these functionals
correspond to (i.e., are mathematical models of) a thin (but not
infinitely thin) knotted resilient wire constrained between two
parallel planes.

The idea of defining energy functionals for classical knots is due
to H.~K.~Moffat~\cite{Mof}. It was further developed by
W.~Fukuhara~\cite{Fuku}, J.~O'Hara in~\cite{O-H1}, \cite{O-H2},
\cite{O-H3}, \cite{O-H4}, M.~H.~Freedman, and  Z.-X.~He
in~\cite{Freed2}, M.~H.~Freedman, Z.-X.~He, and Z.~Wang
in~\cite{Freed} and together with S.~Bryson in~\cite{Freed3},
D.~Kim, R.~Kusner in~\cite{Kim}, O.~Karpenkov in~\cite{EKar1},
\cite{EKar2} and others. The main underlying idea is to define a
real-valued functional on the space of knots so that gradient
descent along the values of the functional does not change the
ambient isotopy class of the knot and leads to a well-defined
minimum, which may be regarded as the normal form of the knot
(corresponding to the given functional); if this gradient descent
leads to the same normal form starting from any two knots in the
same ambient isotopy class, we would have a solution of the knot
classification problem -- knots would be entirely classified by
their normal forms.

However, none of the functionals studied so far have achieved this
ambitious goal, and most of the experts appear to be skeptical
about the existence of such a functional. Nevertheless, the
functionals considered, in particular the M\"obius energy
functional devised by J.~O'Hara, possess some important
properties,  for example, computer experiments show (see
\cite{Freed}) that any knot diagram of the unknot is taken to the
round circle via gradient decent along the values of the M\"obius
energy functional, but the claim that the unknot has a unique
normal form (the circle) has not been proved.

\smallskip
The functionals studied in the present paper differ from those
mentioned above in that they are defined on knot diagrams (rather
than knots in 3-space), gradient descent for them is not invariant
with respect to the Reidemeister $\Omega_1$ move (but is invariant
with respect to $\Omega_2$ and  $\Omega_3$), and the diagrams
themselves are not curved lines (they are like very thin solid
tori lying on the plane). The real life prototypes of these ``thin
solid flat knots" were resilient thin wires constrained between
two parallel planes. They were studied experimentally by
A.~B.~Sossinsky in \cite{ABS1}.

Consider a tubular flat knot of length $2\pi$ and
cross-section radius $\varepsilon>0$, $\varepsilon\ll 1$ (the last
inequality means ``as small as we need in the given situation'');
the knot is constrained between two parallel planes at a small
distance (say $6\varepsilon$) from each other. Let us supply
the core $K_0$ (central line) of our knot $K$ with an energy
functional $E$ consisting of two summands:
$$
E=U+R.
$$
The functional $U$, which we shall call a {\it uniformizational}
functional, serves to bring the core (via gradient descent)  to a
``perfectly rounded off'' form. In this paper we consider a family
of functionals depending on arbitrary continuous function $f$,
$$
U_f(\gamma)=\int\limits_{S^1}f(\kappa(\gamma(t)))dt,
$$
where $\kappa(\gamma(t))$ is the oriented curvature at the point
$\gamma(t)$. The summand $R$ (which we call the {\it resistance}
functional)  forbids to perform deformations under which the knot
type of the diagram changes (in particular, it forbids crossing
changes). We define it as follows
$$
R(\gamma)=\sum\limits_{\ell\in \alt(\gamma)} \frac{1}{A(\ell)},
$$
where the sum is taken over all cycles in knot diagrams bounding a
disk and such that in their consecutive  intersections overpasses and underpasses alternate  (we refer to Subsection~2.1 for the exact definitions), and by $A(\ell)$ we denote the area of the cycle $\ell$. Further (in Subsections~2.2 and~2.3), we introduce two modification of this functional that are easy to compute.

\vspace{2mm}

This paper is organized as follows.

In the first section, we study the uniformization functionals. In
Subsection~1.1, we define the family of uniformization functionals
$U_f(\cdot)$. In Subsection 1.3, we show how to calculate the
minima of these functionals by using the Gauss representation of
regular curves described in Subsection 1.2. In Subsection 1.4, we
perform such calculations for the particular case in which the
function $f$ is given by $f(x)=x^2$. It turns out that the
critical curves are related to the trajectories of pendulum. We
prove that all the critical trajectories are either circles passed
several times, or $\infty$-shaped curves (defined by elliptic
functions) passed several times.

Section~2 is dedicated to resistance functionals. We start with the
general notion of resistance energy in Subsection~2.1. In
Subsections~2.2 and 2.3, we show how to simplify the complexity of
resistance energy (material resistance energy and its genericity
modification). In Subsection~2.4, we show that a regular
deformation of bounded genericity energy preserves the knot type.
Finally, in Subsection~2.5, we say a few words about the influence
of material resistance energy on the shape of the normal form in
the pendulum case (i.e., for $f(x)=x^2$).

In the brief concluding Section 3, we discuss the perspectives of this
research and present some conjectures.

\section{On the uniformization of regular curves}

In this section, instead of knot diagrams, we consider regular
$C^2$-curves of constant length $2\pi$ in the plane, in particular
regular curves, construct various functionals on the space of such
curves and investigate the normal forms given by the functionals.
Gradient descent along  these functionals tends to make the local
structure of the curves more uniform (in a sense related to their
curvature), and so we call them uniformization functionals.

\subsection{A general family of uniformization functionals}

Let $f$ be a real continuous function satisfying $f(0)=0$. A {\it
uniformization functional $U_f$} associated with the function $f$
is the functional defined at a regular $C^2$-curve $\gamma$ (with
arclength parameter $t$) as follows
$$
U_f(\gamma)=\int\limits_{S^1}f(\kappa(\gamma(t)))dt,
$$
where $\kappa(\gamma(t))$ is the oriented curvature at the point
$\gamma(t)$.

\begin{example}
In the case when $f(x)=x$, the resulting functional is constant on
regular homotopy classes of regular $C^2$-curves and equal to the
Whitney index of the curve, i.e., the number of revolutions performed by the tangent vector to the curve as its initial point goes once around the curve.

\end{example}

Let us expand the definition of uniformization functionals to the
broader class of curves for which $\kappa(\gamma(t))$ is not defined
at all points of the given curve.

\begin{definition}
Let $f$ be a positive continuous function satisfying $f(0)=0$. The
{\it extended uniformization functional $\hat U_f$} associated with the function $f$ is the functional defined at a regular $C^2$-curve  $\gamma$ (with arc length parameter $t$) as follows
$$
\hat U_f(\gamma)=\lim\limits_{\varepsilon\to 0}
\frac{1}{\varepsilon}\int\limits_{S^1}f\bigg(\frac{1}
{R\big(\gamma(t{-}\varepsilon),\gamma(t),\gamma(t{+}\varepsilon)
\big)}\bigg)dt,
$$
where $R(\gamma(t-\varepsilon),\gamma(t),\gamma(t+\varepsilon))$
is the radius of the circle passing through the points
$\gamma(t-\varepsilon)$, $\gamma(t)$, and $\gamma(t+\varepsilon)$.
\end{definition}

The next assertion readily follows from compactness arguments.

\begin{proposition}
For any $C^2$-regular curve $\gamma$,
$$
\hat U_f(\gamma)=U_f(\gamma).
$$
\qed
\end{proposition}

It is natural to suppose that the first interesting family of
functionals for further study is the class associated with functions
of the form $x^\alpha$ for $\alpha>1$ (we shall say a few words
later about the case  $\alpha=2$ and its relationship to pendulum
motions).

\subsection{Gauss representation of regular curves}

Consider a regular curve $\gamma$ of length $2\pi$ with arclength
parametrization $t$. Consider a function $\alpha:[0,2\pi]\to\r$
such that
$$
\dot\gamma(t)=(\cos\alpha(t),\sin\alpha(t)).
$$
We then say that $\alpha$ is a {\it Gauss representation} of the
regular curve $\gamma$. (Here and later by $\dot g$ we denote the
derivative ${\partial g}/{\partial t}$.)

Suppose now that we are given a function $\alpha$. {\it When does
this function represent a regular curve?} Let us briefly answer
this question. Consider the following three conditions on
$\alpha$. Since the curve is closed, we have the ''coordinate''
conditions:
\begin{equation}\label{cond1}
\int\limits_0^{2\pi}\cos\alpha(t)dt=\int\limits_0^{2\pi}\sin\alpha(t)dt=0.
\end{equation}
In addition we know that
\begin{equation}\label{cond2}
\alpha(0)=\alpha(2\pi).
\end{equation}

\begin{proposition}
If the above three conditions are satisfied, then the resulting
curve is regular. \qed
\end{proposition}

\subsection{Critical points of the energy functionals $U_f$}

Let us study critical points of the energy functionals $U_f$. Notice
that in the Gauss representation we have $\kappa=\dot \alpha$, so
the functional is of the following form:
$$
U_f(\alpha)=\int\limits_0^{2\pi}f(\dot \alpha) dt.
$$

\begin{theorem}
Consider a $C^2$-regular curve $\gamma$ and a nonnegative
$C^2$-function $f$ satisfying $f(0)=0$. Let $\alpha$ be the Gauss
representation of $\gamma$. Suppose $\alpha$ is critical for
$U_f(\gamma)$. Then there exist constants $C_1$ and $C_2$ such
that $\alpha$ satisfies
\begin{equation}\label{e1}
f''(\dot \alpha)\ddot \alpha=C_1\cos \alpha+C_2\sin\alpha.
\end{equation}
\end{theorem}

\begin{proof}
Consider a variation $\alpha+h\beta$ with a small parameter $h$.
First, notice the following. As we vary the closed curve, the
variation $\alpha+h\beta$ satisfies conditions~(\ref{cond1})
above, i.e., the derivative $\frac{d}{dh}$ of the corresponding
integral equals zero, which is equivalent to
\begin{equation}\label{e2}
\int\limits_0^{2\pi}\sin(\alpha(t))\beta(t)dt=
\int\limits_0^{2\pi}\cos(\alpha(t))\beta(t)dt=0.
\end{equation}
Secondly, for the case in which $\alpha$ is a critical point, all the
variations are zero. Hence
$$
\int\limits_0^{2\pi}f'(\dot \alpha)\dot \beta dt=0, \quad
\text{which is equivalent to} \quad
\int\limits_0^{2\pi}f''(\dot \alpha)\ddot \alpha \beta dt=0.
$$
The last equation holds for any variation $\beta$ satisfying
Equations (4);  
therefore, we obtain
$$
f''(\dot \alpha)\ddot \alpha=C_1\cos\alpha+C_2\sin\alpha,
$$
where $C_1$ and $C_2$ are some constants.
\end{proof}

\begin{remark}
Theorem 1.5 means that the minima of the functional $U_f$ may be
found by solving (for $\alpha$) the differential equation (1) with
the appropriate initial conditions. Then the values of $\alpha$ at
such minima specify curves that may be regarded as normal forms of
the given curve. This approach allows to obtain normal forms by
direct computations not involving gradient descent along the
values of the functional.
\end{remark}

\subsection{The case $f(x)=x^2$: the pendulum}

 Let $f(x)=x^2$, then
Equation~(3) 
becomes
$$
\ddot \alpha=C_1\cos\alpha+C_2\sin\alpha.
$$
After an appropriate Euclidean transformation, we obtain the
equation for the simple pendulum
$$
\ddot \alpha+\omega^2\sin\alpha=0
$$
for some nonnegative constant $\omega$ whose physical meaning is
as follows:
$$
\omega=\sqrt{g/L};
$$
here $L$ is the length of the pendulum, and $g$ is the
gravitational acceleration.

\vspace{2mm}

There are trivial solutions whose trajectories are circles passed
several times. They correspond to $\omega=0$. The others solutions
are more complicated. Consider the case $\omega> 0$. The law of
the motion of the pendulum is described by the equation
$$
\sin\frac{\alpha}{2}=\xi\sn(\omega t+t_0|\xi),
$$
where $\sn$ is the Jacobi elliptic sine, whose value $\sn(u|k)$ is
defined from the equation
$$
\sn(u|m)=\sin \phi,
$$
where $\phi$ is a solution of the equation
$$
u=\int\limits_0^{\phi} \frac{dt}{\sqrt{1-m^2\sin^2 t)}}.
$$
The constants $t_0$ and $\xi$ are parameters. Basically $t_0$ is
the starting time parameter on the curve, so without loss of
generality, we put $t_0=0$. The parameters $\omega$ and $\xi$
essentially define the trajectory.

Let us answer the following question.

\vspace{2mm}

{\bf Question.} {\it  Which trajectories $\{\alpha(t)|0\le t\le
2\pi\}$ defined by $(\xi,\omega)$ result in a closed
differentiable curve $\gamma_\alpha$?}

\vspace{2mm}

The answer to this question is as follows.

\begin{theorem}\label{shapes}
All critical curves of the functional are either circles passed
several times, or $\infty$-shaped curves passed several times.
$$
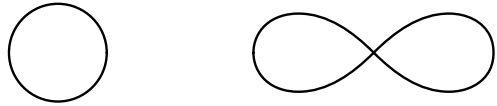
\epsfbox{solution.3}
$$
\end{theorem}

Let $K$ be the complete elliptic integral of the first kind, i.e.,
$$
K(m)=\int\limits_0^1\frac{dt}{\sqrt{(1-t^2)(1-m^2t^2)}}.
$$

\begin{conjecture}\label{conj}
We conjecture that all $\infty$-shaped extremal curves are
homothetic to the curve $\gamma_\alpha$ for $\alpha$ satisfying
$$
\sin\frac{\alpha(t)}{2}=\xi\sn\Big(\frac{2rK(\xi)t}{\pi}\Big|\xi\Big),
$$
where $r$ is a nonzero integer and $\xi\approx .90890856$ $($see
Figure~\ref{solution.1} (left)$)$.
\end{conjecture}

We say a few words related to this conjecture in
Remark~\ref{uniqueness}.

\subsubsection{Proof of Theorem~\ref{shapes}
and construction of the $\infty$-shaped critical curve}

Consider the phase portrait (with coordinates
$(\alpha,\dot\alpha)$) of a pendulum; we do not specify $\omega$,
since all the phase spaces are similar (see  Figure 1). 
There are two non-singular types of trajectories and three
singular types.

\begin{figure}
$$
\epsfbox{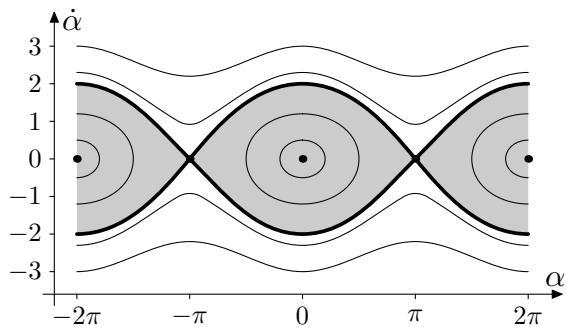}
$$
\caption{The phase portrait of the pendulum.}\label{portrait.1}
\end{figure}


\vspace{1mm}

{\bf Singular case.} In the singular case, we have singular points
of type $(2k\pi,0)$ (stable equilibrium) and $((2k{+}1)\pi,0)$
(unstable equilibrium) and all separatrices (nonperiodic motion of
the pendulum). At the equilibrium points the angle $\alpha$ does
not change, and hence $\gamma_\alpha$ is not closed. In the case
of the separatrices we have
$$
0<|\alpha(0)-\alpha(2\pi)|<2\pi.
$$
Hence the corresponding curve is not differentiable at
$\gamma_\alpha(0)=\gamma_\alpha(2\pi)$ even if it is closed. So in
the singular case there are  no closed differential curves $\gamma_\alpha$.

\vspace{1mm}

{\bf General case of full swing pendulum trajectories.} The full
swing pendulum trajectories are situated in the white regions of
Figure 1.
Let us show that in this case there are also no differentiable closed curves $\gamma_\alpha$.

Without loss of generality, we assume  that $\alpha(0)=0$, and,
therefore, $\alpha(2\pi)=2k\pi$ for some integer $k$ (since the
tangent vectors at $t=0$ and $t=2\pi$ must coincide). It is quite
easy to see that on the segments $[(2m+1/2)\pi,(2m+3/2)\pi]$, the
value of the modulus $|\dot \alpha|$ is smaller than at the
segments $[(2m-1/2)\pi,(2m+1/2)\pi]$. Hence the integral
$$
\int\limits_0^{2\pi}\cos\alpha(t)dt
$$
is nonzero, and therefore  any such curve is not closed.

\vspace{1mm}

{\bf Swing pendulum trajectories that are not full.} Such
trajectories are represented by the ovals on the phase portrait in
the gray regions of Figure 1.
In this case we have several closed differential curves $\gamma_\alpha$.

Again without loss of generality, we suppose that $\alpha(0)=0$,
and, therefore, $\alpha(2\pi)=0$ (since the pendulum never reaches
the vertical position corresponding to unstable equilibrium). This
is equivalent to the relation

$$
\xi\sn(2\pi\omega|\xi)=\sin\frac{\alpha(2\pi)}{2}=0.
$$
The solutions of the equation $\sn(a|b)=0$ are described as follows:
$$
a=2rK(b)+2IsK(1{-}b), \qquad \hbox{for $r,s\in \z$,}
$$
where $I=\sqrt{-1}$.

We conjecture that there are no closed curves in the cases for which
$s\ne 0$.

\begin{remark}\label{uniqueness}
In the major cases, when $s\ne 0$, we end up with solutions
$a=a(r,s,b)$ that have nonzero imaginary parts. Surprisingly, for
the case $r=s$ and $b=3$, the value of $a$, i.e.,
$$
2r(K(2)+IK(3)),
$$
seems to be integer (numerical calculations only show that the
first 1000 digits after the decimal point are zero ). We
conjecture that these cases occur not often enough to have yield
new closed curves $\gamma_\alpha$.
\end{remark}

\medskip
Let us study the case $s=0$. Then
$$
2\pi\omega=2rK(\xi), \qquad \hbox{for $r,\in \z$,}
$$
and so  $\omega=rK(\xi)/\pi$. Now let us check the
conditions for the curve to be closed. The first one is
$$
\int\limits_0^{2\pi}\sin\alpha(t)dt=0.
$$
It turns out that this condition holds automatically for even values of $r$
and does not hold for odd $r$:
$$
\int\limits_0^{2\pi}\sin\alpha(t)dt=-\frac{1}{\omega^2}
\int\limits_0^{2\pi}\ddot\alpha(t)dt=-\frac{1}{\omega^2}\big(\dot\alpha(2\pi)-\dot\alpha(0)\big).
$$
It is clear that $|\dot \alpha (0)|=|\dot\alpha(2\pi)|$, since
$\alpha(0)=\alpha(2\pi)=0$ is the minimal value of the potential
energy of the pendulum. For the signs, we have
$$
\dot\alpha(2\pi)=
\left\{
\begin{array}{l}
\dot\alpha(0) \quad \hbox{if $r$ is even,}\\
-\dot\alpha(0) \quad \hbox{if $r$ is odd.}\\
\end{array}
\right. .
$$
Now let us study the second condition
$$
\int\limits_0^{2\pi}\cos\alpha(t)dt=0.
$$
We have
$$
\int\limits_0^{2\pi}\cos\alpha(t)dt=
\int\limits_0^{2\pi}1{-}2\sin^2\frac{\alpha(t)}{2}dt=
\int\limits_0^{2\pi}1{-}2\xi^2\sn^2\Big(\frac{r
K(\xi)t}{\pi}\Big|\xi\Big)dt.
$$
Denote this integral by $\Delta x(\xi,r)$.

We are interested in the zeros of $\Delta x(\xi,r)$ on the open
interval $(1,1)$ corresponding to the swing pendulum trajectories
that are not full. Experiments show that the function $\Delta
x(\xi,r)$ does not depend on $r$, it is even and has two zeros.
We show this function on Figure 2 (left). 
Experimentally, the zeros of this function are at
$$
\xi\approx \pm .90890856.
$$
The resulting curve $\gamma_\alpha$ for $r=2$ is shown on
Figure 2 (right). 
The curves for the remaining even $r=2k$ are homothetic to those
in the case  $r=2$, since the curve is passed $k$ times and,
therefore, its length must be $2\pi/k$.

\begin{figure}
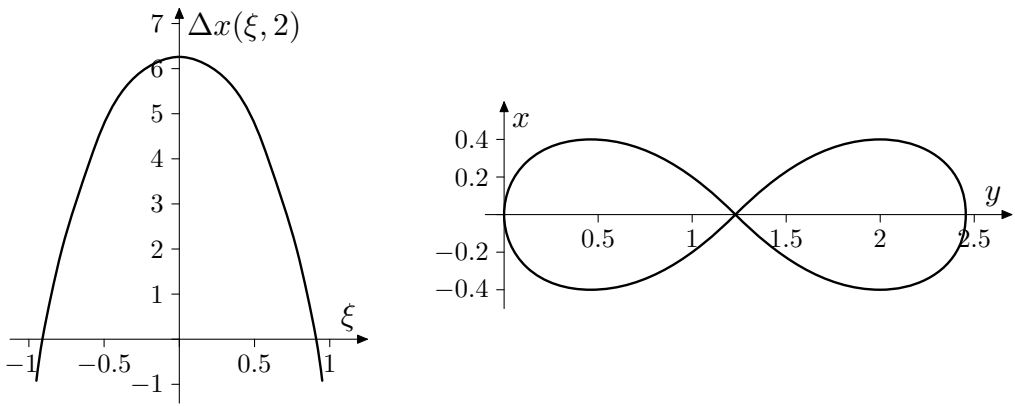

$$
\epsfbox{solution.1}\qquad\epsfbox{solution.2}
$$
\caption{An extremal solution for $r=2$.}\label{solution.1}
\end{figure}



Here are some concluding remarks about uniformization functionals.
For the case of the functionals $f(x)=x^{k}$ for $k>1$, the change
of the Whitney index implies the unboundedness of the energy
$U_f$. And hence the gradient flow of $U_f$ ends up either at some
critical point or at the singular boundary. We conjecture that in
the case $k>1$ the flow never reaches the boundary and, therefore,
it ends up at a critical point (in particular this means that the
Whitney index is preserved by the gradient flow).

In the case of the pendulum (i.e.,~$f(x)=x^{2}$), from
Theorem~\ref{shapes} we obtain the uniqueness of the critical
curve for each class of nonzero Whitney index: the critical curve
is a circle passed several times. In the case of a zero Whitney
index, we have an at least countable number of distinct critical
curves. They are the $\infty$-shaped curves introduced in
Conjecture~\ref{conj}. This case is more complicated, since we do
not know which of these critical curves are stable and which are
not.

\section{Resistance functionals for knot diagrams}

In the previous section, we constructed functionals whose
gradient flow takes diagrams to certain ``perfect" normal forms.
However, during these deformations to normal form, the knot type is
not necessarily  preserved. In this section, we introduce
additional terms that grow to infinity if the deformation changes
the knot type, thus ensuring that the knot type does not change.

\subsection{Resistance energy}

A {\it cycle} $\ell$ in a diagram is a subset of the diagram
homeomorphic to the circle. The {\it area} of a cycle $\ell$ is
the area of the domain it bounds; we denote it by $A(\ell)$.
An arc $ab$ is called {\it alternated} if $a$ is up-going and $b$
is down-going or vice versa ($a$ is down-going and $b$ is up-going).
Here we do not take into consideration intersections with the other
arcs in the interior of $ab$. A cycle is called {\it alternated} if all its
arcs are alternated. Denote the set of alternated cycles in the
diagram $\gamma$ by $\alt(\gamma)$.

\begin{example} In Figure 3 
we show the standard diagram of
the trefoil knot. It has 11 cycles: there are 6 one-arc cycles, 3
two-arc cycles and 2 three-arc cycles. All the cycles are
alternated.
\end{example}

\begin{figure}[h]
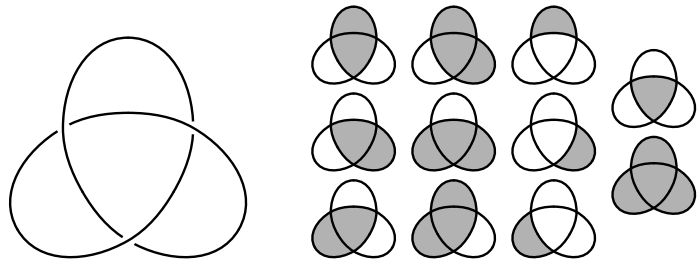

$$
\epsfbox{trefoil.1} \qquad \epsfbox{trefoil.2}
$$
\caption{A trefoil knot diagram and its alternated
cycles.}\label{trefoil.1}
\end{figure}


\begin{definition}\label{redef}
The {\it resistance energy} (or {\it RE} for short) is
$$
RE(\gamma)=\sum\limits_{\ell\in \alt(\gamma)} \frac{1}{A(\ell)}.
$$
\end{definition}

There is a serious practical disadvantage in using resistance
energy. The number of cycles in the diagram grows very rapidly
with respect to the crossing number; in the next two subsections
we show how to fix this problem.

Let us illustrate the situation
with the following example. Consider a graph $G(n)$ isomorphic to the
$n\times n$ lattice as on Figure 3 (left). 

\begin{figure}[h]
$$
\epsfbox{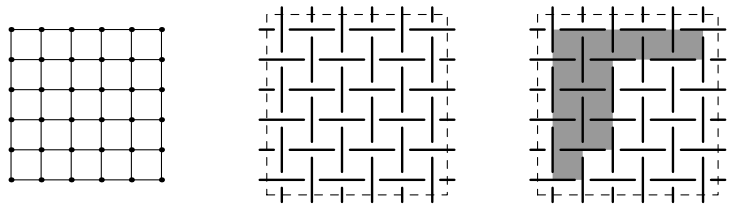}
$$
\caption{The graph $G(5)$, the corresponding alternated diagram
$G^*(5)$, and an alternated cycle representing a Young diagram of
shape $(5,2,2,2,1)$.}\label{g.1}
\end{figure}



The number of cycles in $G(n)$ grows very rapidly with respect to
$n$:

\medskip
\begin{center}
\begin{tabular}{|c||r|r|r|r|r|}
\hline
G(n) &$n=1$&$n=2$&$n=3$&$n=4$&$n=5$\\
\hline \hline
Number of vertices & 4 & 9 & 16 & 25 & 36\\
\hline
Cycles in $G$ & 1 & 13 & 213& 9349& 1222363 \\
\hline
\end{tabular}
\end{center}

\medskip
Suppose our diagram contains an alternated region $G^*(n)$ whose
graph is homeomorphic to $G(n)$ as on Figure 3 (center).

\begin{remark}
A region like $G^*(n)$ often occurs in a slightly perturbed
$\infty$-shaped curve passed several times in the neighborhood of a
former point of self-intersection.
\end{remark}

Now we choose alternated cycles of $G^*(n)$ among all the cycles
of $G(n)$. Let us give a general estimate. The region $G$ has
$(n+1)^2$ double points. For the number of cycles and alternated
cycles in $G(n)$ and in $G^*(n)$, respectively, we have the
following lower bounds.

\begin{proposition}
The number of cycles in $G(n)$ is at least
$$
\frac{4^n(n+1)}{\sqrt{\pi}n^{3/2}}-1.
$$
The number of alternated cycles in $G^*(n)$ is at least
$$
\frac{2^{n}}{\sqrt{2\pi n}}-1.
$$
\end{proposition}

\begin{proof}
The number of nonempty disks in the shape of Young diagrams in $G(n)$
(see an example in Figure 3 (right)
is exactly ${2n\choose n}-1$. The boundaries of these disks are cycles. A simple estimate shows that  at least
${n \choose \lfloor n/2\rfloor}-1$
of the corresponding cycles in $G^*(n)$ are alternated. The
asymptotics for Catalan numbers is
$$
C_n\approx \frac{4^n}{\sqrt{\pi}n^{3/2}}.
$$
This implies the estimates of the proposition.
\end{proof}

\begin{corollary}
The asymptotical growth of the maximal number of alternated cycles
of diagrams with $N$ double points is at least
$C\frac{2^{n/2}}{\sqrt{n}}$ for some constant $C$.
\end{corollary}

In fact, the real growth should be even higher. From the computational
point of view, such a growth rate is not acceptable. In the next
subsections, we describe two ways of reducing the running time.

\smallskip
We conclude this subsection with an interesting geometric
probabilistic problem that arises here. Denote by $D_m$ be a
two-dimensional disk with $m$ holes. Consider a graph $G\subset
\r^2$. Denote by $\# (G, D_m)$ the number of closed subsets of
$\r^2$ homeomorphic to $D_m$ with boundary in $G$. Note that the
number of cycles in a planar graph coincides with $\# (G,D_0)$.
The natural problem here is to {\it find the distributions of such
numbers with respect to the crossing number of $G$.} In addition,
we have the following particular asymptotic problem for $G(n)$.

\begin{problem}
For nonnegative integers $m_1$ and $m_2$ find the following limit
$$
\lim\limits_{n\to\infty}\frac{\# (G(n), D_{m_1})}{\# (G(n),
D_{m_2})}.
$$
\end{problem}

\subsection{Material resistance energy}

Consider a knot made of wire lying on the plane. If the area of
alternated cycles is large enough, then no resistance arises when
we move the wire. Resistance arises only when the area of some
alternated cycle is comparable with the width of the wire. We
formalize this in the following definition. A cycle is called {\it
$\delta$-critical} if its area is less than $\delta$. Denote the
set of all alternated $\delta$-critical cycles of a diagram
$\gamma$ by $\alt_\delta(\gamma)$.

\begin{definition}
The {\it material resistance energy} (or $MRE_\delta$ for short)
is
$$
MRE_\delta(\gamma)=\sum\limits_{\gamma\in \alt_\delta(\gamma)}
\left(\frac{1}{A(\gamma)}-\frac{1}{\delta}\right).
$$
\end{definition}

In general, the parameter $\delta$ depends on the shape (width,
i.e., on $\varepsilon$) of the wire. In the situation of ideal wire
(of zero width, i.e., when $\delta$ tends to 0), the support of
the functional $MRE_0$ is contained in the set of all singular
knot diagrams.

\medskip
The main advantage in this approach is that generically there are
not too many alternated cycles whose areas are very small
(basically tending to 0), so the majority of cycles are not in
$\alt_\delta(\gamma)$. One may determine the set
$\alt_\delta(\gamma)$ in the following two steps.

{\it Step 1}. Find the  ``low-area'' domains containing all simple
cycles for which the area is less than $\delta$ (we say that a
cycle is {\it simple} if it does not contain other points of the
diagram in the interior).

{\it Step 2}. For each connected component of each low-area
domain, check all the alternated cycles contained in its closure.

\begin{remark}
The MRE-functional simulates the real situation with a thin wire,
for which it is impossible to have infinitely many cycles in
diagrams, hence in real life the value of the MRE-functional is
always finite. Anyway, the corresponding algorithms usually work with a
piecewise discretization of knots, and for polygonal knots the
amount of cycles is always finite.
\end{remark}

\subsection{Genericity modification of MRE}

The MRE-functional still has a disadvantage: small-area regions
like $G^*(n)$ may occur at some moment of the gradient flow. This
will lead to an enormous increase in the complexity of the
gradient calculation. In this subsection, we modify the
MRE-functional in such a way that one needs to study only loops,
2-gonal, triangular, and quadrangular cycles.

Denote by $\crit_\delta^4(\gamma)$ the set of all 4-arc cycles in
$\gamma$, denote also by $\alt_\delta^{\le 3}(\gamma)$ the set of
all alternated loops, 2-, and 3-cycles. Let
$$
\Gamma_\delta(\gamma)=\alt_\delta^{\le
3}(\gamma)\cup\crit_\delta^4(\gamma).
$$
\begin{definition}
The {\it genericity material resistance energy} (or $GMRE_\delta$
for short) is
$$
GMRE_\delta(\gamma)=\sum\limits_{\ell\in\Gamma_\delta(\gamma)}
\left(\frac{1}{A(\ell)}-\frac{1}{\delta}\right).
$$
\end{definition}

Any diagram with $n$ crossings has less then ${n^p}/{p!}$ cycles
with $p$ arcs each. Hence the amount of cycles in
$\Gamma_\delta(D)$ is always bounded by
$$
\frac{n^4}{24}+\frac{n^3}{6}+\frac{n^2}{2}+n.
$$

\subsection{On the deformation of knot diagrams}

Let $\gamma:S^1\to \r^2$ be a $C^2$-regular curve. We say that
$\gamma$ is {\it immersed} if the image $\gamma(S^1)$ has finitely
many double points of transversal self-intersections. We say that
$\gamma$ is {\it codimension-one-immersed} if it has finitely many
double points of transversal self-intersections and either a
triple point of transversal self-intersection or a double point at
which the intersection is not transversal.

Let $\gamma_t$ be a deformation in the space of regular curves of
$S^1$ into $\r^2$. We say that a deformation $g_t$ is {\it
generic} if for all but a finite number of points $t$ the image
$g_t(S^1)$ is immersed, and at these points it is
codimension-one-immersed.

We consider a knot diagram as an immersed curve equipped with a
crossing type (i.e., overpass-underpass information) at each
double point. We also extend the notion of knot diagram to the
case of codimension-one-immersed curves. From computational point
of view it is worthy to preserve the crossings (at least locally
for triple points). So we regard a triple point as three double
points and specify crossing types at all these double points
(without taking care whether such three crossings are realizable
as the projection of a knot or not). A tangential double point is
regarded as a pair of double points of the same crossing type.
Denote the space of all immersed and codimension-one-immersed
diagrams by $\D$.

A deformation $\gamma_t$  in $\D$ is {\it generic} if for all but
a finite number of points $t$ the image $g_t(S^1)$ is immersed, at
these points it is codimension-one-immersed, and the crossing
types depend smoothly on the crossings with respect to the
parameter $t$.

\begin{remark}
Notice that a generic deformation of a diagram (in our definition)
does not necessarily preserve the knot type. It can change while
passing through a codimension-one-regular diagram with a triple
point, like it does on the following picture.
$$
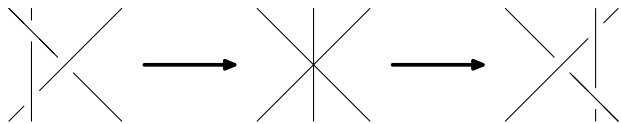
\epsfbox{change.1}
$$

\end{remark}

\begin{theorem}
Let $\gamma_t$ be a generic deformation in the space of regular
knot diagrams with parameter $t$ in $[0,1]$. Suppose that there
exists $\delta$ such that the $GMRE_\delta$-functional is
uniformly bounded from above for the whole deformation. Then the
knot types of $\gamma_0$ and of $\gamma_1$ coincide.
\end{theorem}

\begin{proof}
Since the deformation is generic, the diagrams occurring in the
deformation have only double tangency points or triple points. The
surgeries at these points are described by the second and the
third Reidemeister moves. All other surgeries are forbidden,
otherwise an alternated cycle disappears (but its area is bounded
from below by $1/C$).
\end{proof}

\subsection{Total energy related to the pendulum}
Let us say a few words about the combination (sum) of the
uniformization functional $U_f$ for $f(x)=x^2$ and the material
resistance energy $MRE$, i.e.,
$$
E(\gamma)=U_{x^2}(\gamma) + MRE_\delta(\gamma).
$$
Critical positions of $U_{x^2}$ are either circles passed several
times or $\infty$-shaped curves passed several times. In the case
of a nonzero Whitney index, the functional $MRE_\delta$ does not
allow the knot diagram to reach the shape of perfect circles or
perfect $\infty$-shaped curves. The knot diagram ends up either as
a circular braid (which is almost a circle passed several times)
or as almost an $\infty$-shaped curve passed several times. In
addition, it is possible that the gradient flow will take the knot
diagram to a diagram containing alternated cycles of small area.
For instance, this is the case when the knot type is non-trivial
and the Whitney index of the starting curve is one. So we have the
following

\begin{theorem}
The critical knots for the energy
$$
E(\gamma)=MRE_\delta(\gamma)+U_{x^2}(\gamma)
$$
are either a circle passed once, or an $\infty$-shaped curve
passed once, or contains at least one alternated cycle of area
less than $\delta$.
\end{theorem}

\begin{proof}
In Theorem~\ref{shapes}, we showed that the critical curves for
$U_{x^2}$ are exactly circles and $\infty$-shaped curves. For all
other critical points, the functional $MRE_\delta$ must
contribute, which implies that at least one of the critical cycles
is of area less then $\delta$.
\end{proof}

The same statement holds for the repulsive energy $GMRE_\delta$.

\begin{theorem}
The critical knots for the energy
$$
E(\gamma)=GMRE_\delta(\gamma)+U_{x^2}(\gamma)
$$
are either a circle passed once, or an $\infty$-shaped curve
passed once, or contains at least one alternated cycle of area
less than $\delta$. \qed
\end{theorem}

\section{Perspectives}

In the sequel to this paper, we intend to describe the results of
computer experiments with the functionals defined above, or, more
precisely, with discretizations of these functionals. We
conjecture that the experiments will show that the normal forms
obtained classify flat $\varepsilon$-knots up to flat isotopy, but
only when the number $c$ of crossings is small (say $c\leq 30$).

\bigskip
A further potentially interesting direction of study is related to the
application of Theorem 1 from [12], which asserts that there is a deep
relationship between normal forms of flat $\varepsilon$-knots and
those of ordinary (3D) knots, in order to obtain information on the
latter by using results obtained for the former.

\vspace{5mm}
\end{document}